\documentclass{article}

\usepackage[T1]{fontenc}

\usepackage{amsmath}
\usepackage{amssymb}
\usepackage{amsthm}
\usepackage{mathtools}
\usepackage{graphicx}
\usepackage{xcolor}
\usepackage{hyperref}
\usepackage[square,sectionbib]{natbib}
\usepackage{algorithm}
\usepackage{subcaption}
\usepackage{booktabs}
\usepackage[noend]{algpseudocode}

\newtheorem{theorem}{Theorem}

\begin{document}

\title{Analyzing Brain Structural Connectivity as Continuous Random Functions}
\author{William Consagra \and
Martin Cole \and 
Zhengwu Zhang}

\author{William Consagra\\
  Department of Biostatistics and Computational Biology \\
  University of Rochester, Rochester, NY, U.S.A. \vspace{.2cm} \\
  Martin Cole \\
  Department of Biostatistics and Computational Biology \\
  University of Rochester, Rochester, NY, U.S.A. \vspace{.2cm} \\
  Zhengwu Zhang \\
  Department of Statistics and Opertaions Research \\ 
  UNC Chapel Hill, Chapel Hill, NC, U.S.A.
}

\maketitle

\begin{abstract}
This work considers a continuous framework to characterize the population-level variability of structural connectivity. Our framework assumes the observed white matter fiber tract endpoints are driven by a latent random function defined over a product manifold domain. To overcome the computational challenges of analyzing such complex latent functions, we develop an efficient algorithm to construct a data-driven reduced-rank function space to represent the latent continuous connectivity. Using real data from the Human Connectome Project, we show that our method outperforms state-of-the-art approaches applied to the traditional atlas-based structural connectivity matrices on connectivity analysis tasks of interest. We also demonstrate how our method can be used to identify localized regions and connectivity patterns on the cortical surface associated with significant group differences. Code will be made available at \href{https://github.com/sbci-brain}{https://github.com/sbci-brain}.
\end{abstract}
\noindent%
\vfill

\section{Introduction}
The structural connectivity (SC) describes connectivity patterns between regions of the brain generated by the white matter (WM) fiber tracking results from diffusion-weighted MRI (dMRI). Understanding the heterogeneity of SC across individuals and its relationship to various traits of interest is of fundamental importance for understanding the brain. The most common approach for analyzing SC data is to represent it as a symmetric adjacency matrix, i.e., a network, denoted $\boldsymbol{A}$, referred to as the (structural) connectome matrix. To obtain $\boldsymbol{A}$, the brain surface is discretized into $V$ disjoint regions of interest (ROIs) using some predefined parcellation \citep{desikan2006,destrieux2010}. The element $\boldsymbol{A}_{ab}$ quantifies the extent of structural connectivity between ROIs $a$ and $b$. Joint analysis of a sample of SC from multiple subjects can be performed utilizing network analysis models. For example, in \citep{durante2017nonparametric,zhang2019,arroyo2021b}, the proposed techniques learn a shared latent space characterizing the variability of a sample of networks. Subsequent statistical analysis and inference to relate this variability to cognitive and psychiatric traits of interest can then be performed in the latent space.
\par 
The dependence on the pre-specification of an atlas in the discrete ROI-based framework of analyzing SC is problematic for at least two major reasons. First, analyses are known to be sensitive to the choice of atlas \citep{zalesky2010} and therefore studies can draw different conclusions for different parcellation schemes of the same data. Second, the ROIs can be large and thus introduce information loss, since fine-grained connectivity information on the sub-ROI level is aggregated in the construction of the adjacency matrix. Studies indicate that connectivity information at higher resolutions can capture more valuable information for many traits of interest \citep{mansour2021}. 
\par 
A series of recent works \citep{moyer2017,cole2021,mansour2022} have addressed these problems by transitioning from the discrete, ROI-based representation of brain connectivity to a fully continuous model. These works assume the observed pattern of WM fiber tract endpoints is driven by an unobserved continuous function, referred to as the \textit{continuous connectivity}. This function is defined on the product space of the cortical surface with itself and governs the strength of connectivity between any pair of points on the surface. Critically, the continuous model of connectivity does not depend on the pre-specification of an atlas and therefore avoids the previously outlined issues which plague traditional discrete network based approaches. Despite these advantages, the current approaches suffer from significant computational hurdles for multi-subject analysis. This is because the continuous connectivity is represented as an extremely high-dimensional matrix, obtained through discretization over a high-resolution grid on the cortical surface. For instance, the popular surface mesh from \cite{vanessen2012} results in matrices of dimension $\approx 64,000\times 64,000$. Applying the existing joint network analysis approaches to a sample of matrices of this size is prohibitive. 
\par 
In this work, we extend the continuous connectivity framework to explicitly model the population variability of the SC by considering a random intensity function which governs the distribution of the observed WM fiber tract endpoints. Utilizing techniques from functional data analysis (FDA), we develop a novel methodology and accompanying estimation algorithm which learns a data-adaptive reduced-rank function space for parsimonious and efficient data representation, allowing us to overcome the enormous computational challenges involved in analyzing a sample of super high-dimensional connectivity matrices. We apply the proposed method to a sample of Human Connectome Project (HCP) subjects and show that our method outperforms traditional atlas-based approaches on tasks related to identifying group differences in SC associated with cognitive and psychiatric traits.

\section{Methodology}

\textbf{A Model for Continuous Connectivity}: For a sample of $i=1, ..., N$ subjects, let $O_{i}$ be endpoints of the WM fiber tracts connecting cortical surfaces for subject $i$. Let $\mathbb{S}^2_1$ and $\mathbb{S}^2_2$ be two independent copies of the 2-sphere and let $\Omega = \mathbb{S}^2_1 \cup \mathbb{S}^2_1$. Since the the left and right cortical surfaces are diffeomorphic to a 2-sphere, $O_{i}$ can be modeled as a point process on $\Omega\times\Omega$ \citep{moyer2017}. Define the intensity function $U_i: \Omega \times \Omega \mapsto [0,\infty)$, such that $U_i$ is symmetric and $L^2$, i.e. $U_i(\omega_1,\omega_2) = U_i(\omega_2,\omega_1)$ and $\int_\Omega\int_\Omega U_i^2<\infty$. For any two distinct/non-overlapping and Borel measurable regions $E_1 \subset \Omega$ and $E_2 \subset \Omega$, denote $N(E_1,E_2)$ as the counting process of the number of WM fiber tracts ending in $(E_1,E_2)$. Then $U_i$ satisfies 
$\mathbb{E}\left[N(E_1,E_2)\right]=\int_{E_1} \int_{E_2} U_i(\omega_1,\omega_2)d\omega_1d\omega_2 < \infty.$
Said another way, $U_i$ determines, up to first-order moment, the pattern of fiber tract endpoints on the cortical surface. For any pair of points $(\omega_1,\omega_2)\in \Omega\times\Omega$, subject-level estimates of $U_i(\omega_1,\omega_2)$ can be formed from $O_i$ using the point-wise product heat kernel density estimator (KDE) proposed in \cite{moyer2017}.
\par 
In order to properly account for the population variability of the SC, we model the sample of $N$ point patterns using the theory of doubly stochastic point-processes. In this framework, we assume that intensity function $U_i$ governing $O_i$ is itself a realization of an underlying \textit{random function}, which we denote as $U$. We analyze the variability in the SC using the sample of continuous connectivity  $\{U_1,...,U_N\}$.
\par\bigskip
\noindent{\textbf{Reduced-Rank Embedding of a Sample of Continuous Connectivity}:  Due to the infinite dimensionality of the continuous connectivity, we must construct an efficient representation before conducting meaningful statistical analysis. In this section, we propose a data-adaptive reduced-rank space for the joint embedding of a sample of continuous connectivity.}
\par 
Since we are interested in characterizing the variability of the SC, without loss of generality, assume that $U$ has been centered, i.e. $\mathbb{E}[U]=0$. Define the symmetric separable orthogonal function set of rank $K$:
$$
\mathcal{V}_{K} = \{\xi_k\otimes\xi_k: \xi_k \in L^2(\Omega), \langle \xi_k, \xi_j\rangle_{L^{2}(\Omega)} = \delta_{kj},\text{ for }k = 1,2,...,K\},
$$
where $\xi\otimes\xi(\omega_1,\omega_2) := \xi(\omega_1)\xi(\omega_2)$ and $\delta_{ik}$ is the Kronecker delta. The $i$'th continuous connectivity can be written as a linear combination of basis functions in $\mathcal{V}_K$ plus a residual: $ U_i = \sum_{k=1}^K S_{ik}\xi_k\otimes\xi_k + R_{K,i}$, with $S_{ik} = \langle U_i, \xi_k\otimes\xi_k\rangle_{L^{2}(\Omega\times\Omega)}$ and symmetric residual $R_{K,i}$, which is orthogonal to $\text{span}\left(\mathcal{V}_{K}\right)$. As such, any $U_i$ can be identified with a $K$-dimensional Euclidean vector $\boldsymbol{s}_i := [S_{i1},...,S_{iK}]^T$. The mapping $U_i\mapsto \boldsymbol{s}_i$ is an isometry between $\left(\text{span}(\mathcal{V}_K), \langle\cdot,\cdot\rangle_{L^{2}(\Omega\times\Omega)}\right)$ and $\left(\mathbb{R}^K, \langle\cdot,\cdot\rangle_2\right)$, where $\langle\cdot,\cdot\rangle_2$ is the standard Euclidean metric. Therefore, we can properly embed the continuous connectivity into a $K$-dimensional Euclidean vector space and utilize a multitude of existing tools from multivariate statistics for analysis and inference. 
\par 
To facilitate powerful embeddings, we construct a $\mathcal{V}_{K}$ that is adapted to the distribution of $U$. This is accomplished utilizing a greedy learning procedure. In particular, given a sample of $N$ realizations of $U_i \sim U$, we iteratively construct $\mathcal{V}_K$ by repeating the following steps 
\begin{equation}\label{eqn:empirical_greedy_algorithm}
\begin{aligned}
        \xi_k &= \sup_{\xi\in \mathbb{S}^{\infty}(\Omega)}N^{-1}\sum_{i=1}^N|\left\langle R_{k-1,i}, \xi\otimes\xi\right\rangle_{L^2(\Omega\times\Omega)}|^2 \\
    \quad R_{k,i} &= U_i - P_{\mathcal{V}_{k}}(U_i), \quad \mathcal{V}_k = \mathcal{V}_{k-1} \cup \{\xi\otimes\xi\}\\ 
\end{aligned}
\end{equation}
for $k=1,...,K$, where $P_{\mathcal{V}_{k}}$ is the $L^2$ orthogonal projection operator onto $\text{span}\left(\mathcal{V}_k\right)$, $\mathbb{S}^{\infty}(\Omega) :=  \{\xi \in L^2(\Omega) : \left\| \xi \right\|_{L^{2}(\Omega)} = 1\}$ and the process is initialized with $R_{0,i} = U_i$ and $\mathcal{V}_0=\emptyset$. In order to understand the theoretical performance of our representation space constructed using the updates in \eqref{eqn:empirical_greedy_algorithm}, Theorem~\ref{thrm:sample_convergence_rate} establishes an asymptotic universal approximation property as a function of the rank.
\begin{theorem}\label{thrm:sample_convergence_rate}
Let $U_1, ..., U_N$ be i.i.d. with $U_i \sim U$. Under some minor assumptions on the distribution of $U$, the asymptotic mean residual $L^2$ error resulting from the greedy algorithm defined in \eqref{eqn:empirical_greedy_algorithm} is bounded as  
$$
\lim_{N\rightarrow\infty} N^{-1}\sum_{i=1}^N \left\|R_{K,i}\right\|_{L^2(\Omega\times\Omega)}^2 \le  \frac{B\left(\sum_{k=1}^\infty\sqrt{\rho}_k\right)^2}{K+1}
$$
where $B$ is a finite positive constant related to the smoothness of $U$ and $\rho_k$ is the $k$'th eigenvalue of the covariance operator of $U$.
\end{theorem}
Our functional greedy learning procedure \eqref{eqn:empirical_greedy_algorithm} facilitates efficient computation in multiple ways. It requires solving an optimization problem over $L^2(\Omega)$, opposed to the full space $L^{2}(\Omega\times\Omega)$. Optimization for functions directly in $L^{2}(\Omega\times\Omega)$ square the number of unknown parameters compared to $L^2(\Omega)$, a manifestation of the curse of dimensionality. Additionally, in practice, we are able to utilize the smoothness of elements in $L^2(\Omega)$ to further reduce the dimensionality using basis expansion over $\mathbb{S}^2$.
\par\bigskip
\noindent{\textbf{Deriving the Optimization Problem}: The optimization problem in \eqref{eqn:empirical_greedy_algorithm} is intractable since the search space $\mathbb{S}^{\infty}(\Omega)$ is infinite-dimensional. In practice, we address this by approximating the infinite dimensional parameter $\xi_k$ using basis expansion. Let $\boldsymbol{\phi}_{M_{d}} = (\phi_1, ..., \phi_{M_{d}})^\intercal$ be the spherical splines basis of degree 1 defined over spherical Delaunay triangulation $\mathcal{T}_d$ \citep{lai2007} for $d=1,2$. We form the approximation:}
 $\xi_{k}(\omega) = \boldsymbol{c}_{1,k}^\intercal\boldsymbol{\phi}_{M_{1}}(\omega)\mathbb{I}\{\omega\in\mathbb{S}^2_1\} + \boldsymbol{c}_{2,k}^\intercal\boldsymbol{\phi}_{M_{2}}(\omega)\mathbb{I}\{\omega\in\mathbb{S}^2_2\}$,
where the $\boldsymbol{c}_{d,k} \in \mathbb{R}^{M_{d}}$ are the vectors of coefficients with respect to the basis $\boldsymbol{\phi}_{M_{d}}$. Denote the collected vector of coefficients $\boldsymbol{c}_k = (\boldsymbol{c}_{1,k}^\intercal,\boldsymbol{c}_{2,k}^\intercal)^\intercal$ and basis functions $\boldsymbol{\phi}_M = (\boldsymbol{\phi}_{M_{1}}^\intercal,\boldsymbol{\phi}_{M_{2}}^\intercal)^\intercal$, where $M = M_{1}+M_{2}$. 
\par 
For tractable computation of the $L^2(\Omega\times\Omega)$ inner product, we form a discrete approximation using a dense grid of $n$ points in $\Omega$, denoted as $\boldsymbol{X}$. Define $\boldsymbol{\Phi}\in\mathbb{R}^{n \times M}$ to be the matrix of evaluations of $\boldsymbol{\phi}_{M}$ over $\boldsymbol{X}$ and denote the matrix of $L^2(\Omega)$ inner products of $\boldsymbol{\phi}_{M}$ as $\boldsymbol{J}_{\boldsymbol{\phi}}\in\mathbb{R}^{M\times M}$, with $ \langle\phi_{i},\phi_{j}\rangle_{L^{2}(\Omega)}=0$ if splines $\phi_{i}$ and $\phi_{j}$ are not on the same copy of $\mathbb{S}^2$. The $k$'th discretized residual for the $i$'th subject is approximately $\boldsymbol{R}_{k,i} = \boldsymbol{Y}_i - \sum_{j=1}^{k-1} s_{ij}(\boldsymbol{\Phi}\boldsymbol{c}_j)(\boldsymbol{\Phi}\boldsymbol{c}_j)^\intercal$, where $\boldsymbol{Y}_i$ is the mean-centered KDE estimate of the continuous connectivity over $\boldsymbol{X}$ for the $i$'th subject, the $\boldsymbol{c}_{j}$'s are from the previous $k-1$ selections and 
$
s_{ij} = n^{-2}\left\langle\boldsymbol{R}_{j-1,i},(\boldsymbol{\Phi}\boldsymbol{c}_j)(\boldsymbol{\Phi}\boldsymbol{c}_j)^\intercal\right\rangle_{F}.
$
Introducing the auxiliary variable $\boldsymbol{s}=(s_1,...,s_{N})^\intercal\in\mathbb{R}^{N}$, we propose the following discretized formulation of the optimization problem in \eqref{eqn:empirical_greedy_algorithm}:
\begin{equation}\label{eqn:empirical_optm_reformulated_discretized_AO}
\begin{aligned}
    \hat{\boldsymbol{c}}_k = &\underset{\boldsymbol{c} \in \mathbb{R}^{M}}{\text{ argmax }}\sum_{i=1}^N s_i\left\langle\boldsymbol{R}_{k-1,i}, \boldsymbol{\Phi}\boldsymbol{c}\otimes\boldsymbol{\Phi}\boldsymbol{c}\right\rangle_{F} -  \alpha_1 \boldsymbol{c}^\intercal\boldsymbol{Q}_{\boldsymbol{\phi}}\boldsymbol{c} \\
    \textrm{s.t.} \quad & \boldsymbol{c}^\intercal\boldsymbol{J}_{\phi}\boldsymbol{c} = 1, \boldsymbol{c}^\intercal\boldsymbol{J}_{\phi}\boldsymbol{c}_j  = 0 \text{ for }j=1,2,...,k-1 \\
    &\quad \left\|\boldsymbol{c}\right\|_{0} \le \alpha_{2}, \quad s_i = \langle \boldsymbol{R}_{k-1,i}, \left(\boldsymbol{\Phi}\boldsymbol{c}\right)\otimes\left(\boldsymbol{\Phi}\boldsymbol{c}\right)\rangle_{F}, \text{ for } i=1,...,N.
\end{aligned}
\end{equation}
where $\boldsymbol{Q}_{\boldsymbol{\phi}}$ is a matrix encoding the quadratic variation of the candidate solution $\boldsymbol{c}$ and is used to control the ``roughness'' of the basis functions via tuning parameter $\alpha_1 > 0$. For certain applications it may be of interest to promote the estimation of locally supported $\xi_k$, e.g., in order to localize the effect of the $k$'th basis function on the cortical surface. Owing to the local support of the spherical spline basis functions, this can be done by employing an $l_0$ constraint to encourage sparsity in the $\boldsymbol{c}_k$'s, with $\alpha_2$ controlling the degree of localization. Efficient algorithms exist for evaluating spherical splines, computing their directional derivatives as well as performing integration \cite{lai2007}. As a result, the matrices $\boldsymbol{J}_{\boldsymbol{\phi}}$ and $\boldsymbol{Q}_{\boldsymbol{\phi}}$ can be constructed cheaply.
\par\bigskip
\noindent{\textbf{Algorithm}: We apply an alternating optimization (AO) scheme to form an approximate solution to the problem \eqref{eqn:empirical_optm_reformulated_discretized_AO}, iteratively maximizing $\boldsymbol{c}$ given $\boldsymbol{s}$ and vice versa. Denote the singular value decomposition $\boldsymbol{\Phi} = \boldsymbol{U}\boldsymbol{D}\boldsymbol{V}^\intercal$, and let 
$\boldsymbol{G}_{k,i} := \boldsymbol{U}^\intercal\boldsymbol{R}_{k,i}\boldsymbol{U}$ and define the tensor $\mathcal{G}_{k}$ to be the mode-3 stacking of the $\boldsymbol{G}_{k,i}$.
Under the AO scheme, the update for the block variable $\boldsymbol{c}_k$ at iteration  $t+1$ is given by 
\begin{equation}\label{eqn:ctilde_update}
\begin{aligned}
    \boldsymbol{c}^{(t+1)} = \text{max} &\quad \boldsymbol{c}^\intercal\left[\boldsymbol{V}\boldsymbol{D}\left(\boldsymbol{I} - \boldsymbol{P}_{k-1}\right)\left[\mathcal{G}_{k-1}\times_3\boldsymbol{s}^{(t)} - \alpha_1\boldsymbol{D}^{-1}\boldsymbol{V}^\intercal\boldsymbol{Q}_{\boldsymbol{\phi}}\boldsymbol{V}\boldsymbol{D}^{-1}\right]\left(\boldsymbol{I} - \boldsymbol{P}_{k-1}^{\intercal}\right)\boldsymbol{D}\boldsymbol{V}^\intercal\right]\boldsymbol{c} \\
    & \textrm{s.t.} \quad \boldsymbol{c}^\intercal\boldsymbol{J}_{\boldsymbol{\phi}}\boldsymbol{c} = 1, \quad \left\|\boldsymbol{c}\right\|_{0} \le \alpha_{2}, \\
\end{aligned}
\end{equation}
where $\times_d$ denotes the $d$-mode tensor-matrix multiplication and  \begin{equation}\label{eqn:update_P}
        \boldsymbol{P}_{k-1} = \boldsymbol{D}\boldsymbol{V}^\intercal\boldsymbol{C}_{k-1}\left[\boldsymbol{C}_{k-1}^{\intercal}\boldsymbol{J}_{\boldsymbol{\phi}}\boldsymbol{C}_{k-1}\right]^{-1}\boldsymbol{C}_{k-1}^{\intercal}\boldsymbol{J}_{\boldsymbol{\phi}}\boldsymbol{V}\boldsymbol{D}^{-1}  \quad \boldsymbol{C}_{k-1} := \left[\boldsymbol{c}_1, ..., \boldsymbol{c}_{k-1}\right],
\end{equation}
For $\boldsymbol{s}$, the update at iteration $t+1$ is given in closed form by
\begin{equation}\label{eqn:s_update}
    \boldsymbol{s}^{(t+1)} = \mathcal{G}_{k-1} \times_1 (\boldsymbol{D}\boldsymbol{V}^\intercal\boldsymbol{c}_k^{(t+1)})\times_2 (\boldsymbol{D}\boldsymbol{V}^\intercal\boldsymbol{c}_k^{(t+1)}).
\end{equation}
We apply the fast truncated power iterations method from \cite{yuan2011} to approximate the solution to the sparse eigenvector problem in \eqref{eqn:ctilde_update}. Notice that the tensor $\mathcal{G}_{k}$ is $M\times M \times N$ and thus the complexity of the AO-updates are independent of the number of grid points used in the discretization. Since the computational grid $\boldsymbol{X}$ can be made arbitrarily dense, typically $M \ll n$ which can amount to an enormous savings in computation when compared to an alternative approach of performing tensor decomposition directly on the super high-dimensional discretized continuous connectivity matrices $\boldsymbol{Y}_i$. For example, the tensor decomposition method from \citep{zhang2019} failed when applied directly to the high-resolution connectivity data described in Section~\ref{sec:experiments_and_conclusion} on a HPC cluster with 120GB of RAM.

\section{Experiments and Conclusions}\label{sec:experiments_and_conclusion}

\noindent{\textbf{Dataset Description, Preprocessing, Modeling and Implementation}:
In this work, we use brain imaging data and associated measurements related to various cognitive and psychiatric factors from a sample of 437 female subjects in the HCP young adult cohort, publicly available at \url{https://db.humanconnectome.org}. For each subject, we use both the dMRI and the structural T1-weighted images collected using a 3T Siemens Connectom scanner (Erlangen, Germany). The full imaging acquisition protocol as well as the minimal preprocessing pipeline applied to the dMRI data are given in \cite{glasser2013}. fODFs were estimated from the preprocessed diffusion data using CSD \citep{tournier2007} and then SET \citep{st2018surface} was used to estimate the underlying WM fiber tracts and endpoints on the surface. The T1 images were registered to the dMRI using ANTs \citep{avants2009advanced} and the cortical white surfaces were extracted using Freesurfer. To alleviate the misalignment issue resulting from the joint analysis of $N$ cortical surfaces, we parameterized each using a spherical coordinate via the surface inflation techniques from \cite{fischl1999} and then applied the warping function from freesurfer to bring the endpoints to a template space.}
\par 
To estimate $\mathcal{V}_{K}$, we use a spherical spline basis with $M_1=M_2=410$ over a spherical Delaunay triangulation. In accordance with \cite{cole2021}, the KDE bandwidth was set to be $0.005$ and a grid of $4,121$ points on $\Omega$ was used for discretization. We set $\alpha_1=10^{-8}$ and $\alpha_2=40$ to encourage locally supported basis functions. The rank was selected to be K=$100$ using a threshold criteria on the proportion of variance explained, estimated efficiently by $\left\|\mathcal{G}_{K}\right\|_{F}^2/\left\|\mathcal{G}\right\|_{F}^2$. 
\par 
For comparison, we apply several state-of-the-art network embedding methods, namely Tensor-Network PCA (TN-PCA) \citep{zhang2019}, Multiple Random Dot-Product Graph model (MRDPG) \citep{nielsen2018} and Multiple Adjacency Spectral Embedding method (MASE) \citep{arroyo2021b} to the sample of SC matrices obtained using the Destriex atlas \citep{destrieux2010}. TN-PCA can be applied directly to the streamline endpoint count data, while MRDPG and MASE require a binarization of the connectome matrices. Embedding spaces of K=$100$ were constructed using the publicly available implementations for each method. Analysis was performed using R/4.0.2 and MATLAB/2020b on a Linux machine equipped with a 2.4 GHz Intel Xeon CPU E5-2695 and 120GB of RAM.
\par\bigskip
\noindent{\textbf{Hypothesis Testing}:
We compare the power of the embeddings produced by our continuous connectivity framework, denoted as CC, to the network-embeddings produced by TN-PCA, MDRPG and MASE on a set of two group hypothesis tests. We use a set of 80 measurements of traits spanning the categories of cognition, emotion, and sensory. For each trait, two groups are created by selecting the top 100 and bottom 100 of the 437 HCP females, in terms of their measured trait score. Between-subject pairwise distance matrices were computed in the embedding space for each of the methods and p-values for the two group comparison were formed using the Maximum Mean Discrepancy (MMD) test \citep{gretton2012}. The p-values were corrected for false discovery rate (FDR) control using \cite{benjamini1995}.}
\par 
The four panels in Figure~\ref{fig:global_inference} show the results for each of the embedding methods. The y-axis gives the negative log transformed p-values and the colors indicate significant discoveries under a couple FDR control levels. With a threshold of $\text{FDR}\le 0.05$, the embeddings produced by our method are able to identify 22 significant discoveries, compared to 7 or less for the competitors. These results suggests
potentially large gains in statistical power for detecting group differences when modeling the SC data at much higher resolutions than are commonly being used currently.
\begin{figure}
    \centering
    \includegraphics[scale=0.4]{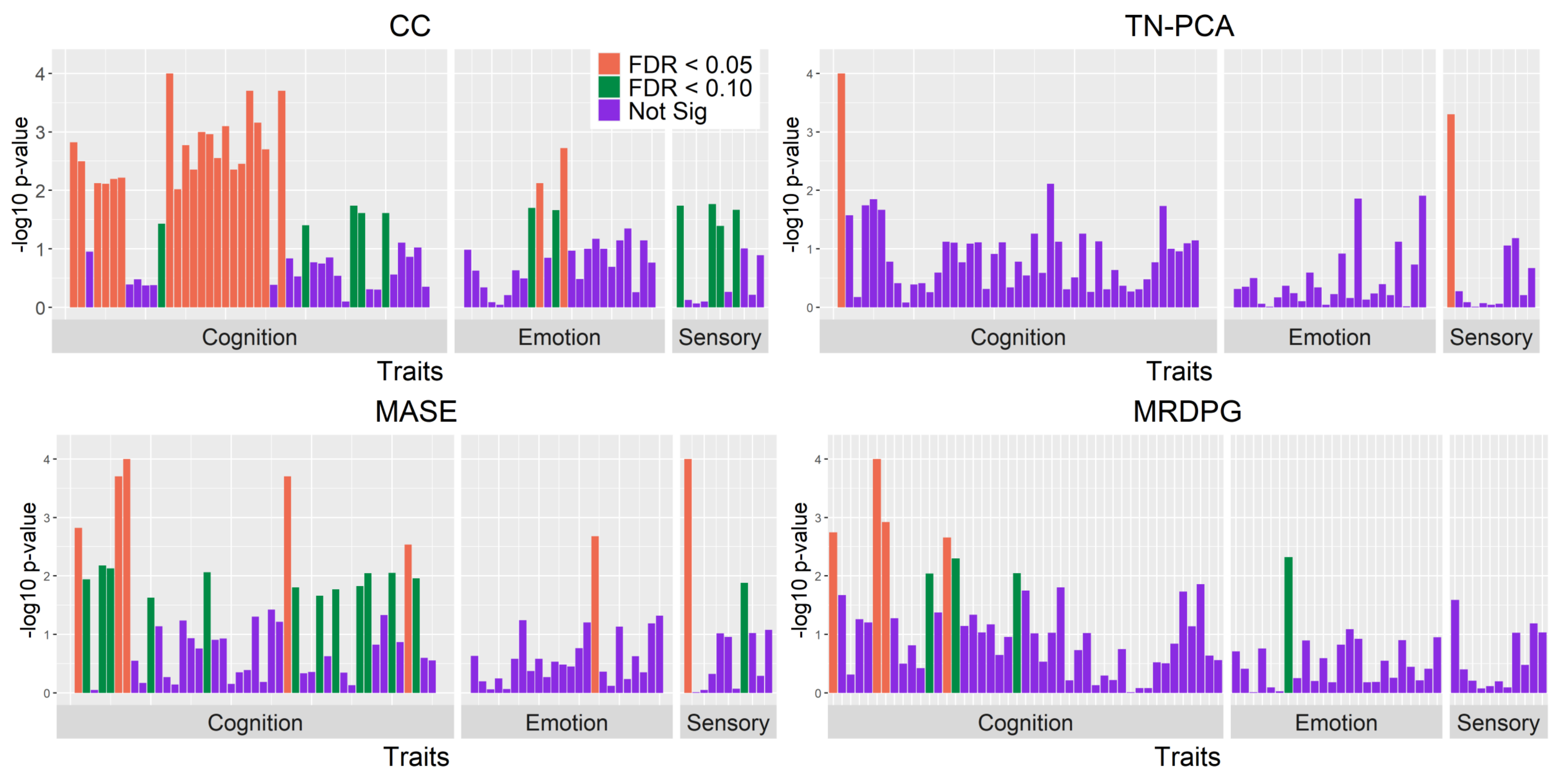}
    \caption{Results from the two-group hypothesis tests for the collection of cognition, emotion and sensory traits. The y-axis gives the $-\text{log}$-transformed p-values.}
    \label{fig:global_inference}
\end{figure}
\par
\noindent{\textbf{Continuous Subnetwork Discovery}: 
We now demonstrate how our method can be utilized for atlas-independent identification of brain regions and connectivity patterns that are related to traits of interest. We consider the cognitive trait \textit{delay discounting}, which refers to the subjective decline in the value of a reward when its arrival is delayed in time and is measured for HCP subjects by the subjective value of $\$200$ at 6 months \cite{estle2006}. We define the steep and low discounting groups using thresholds on this subjective value of $\le \$40$ and $\ge\$160$, respectively. Sub-selecting individuals from our sample who met these criteria resulted in a total sample size of 142: 64 in the steep discounting group and 78 in the low discounting group.  We test for group differences in each embedding dimension using univariate permutation tests on the associated embedding coefficients. 10,000 permutations were used to compute the empirical p-values for each test. A Bonferroni-Holm correction was applied to the p-values to ensure family-wise error rate control $\le 0.10$.}
\begin{figure}
    \centering
    \includegraphics[scale=0.4]{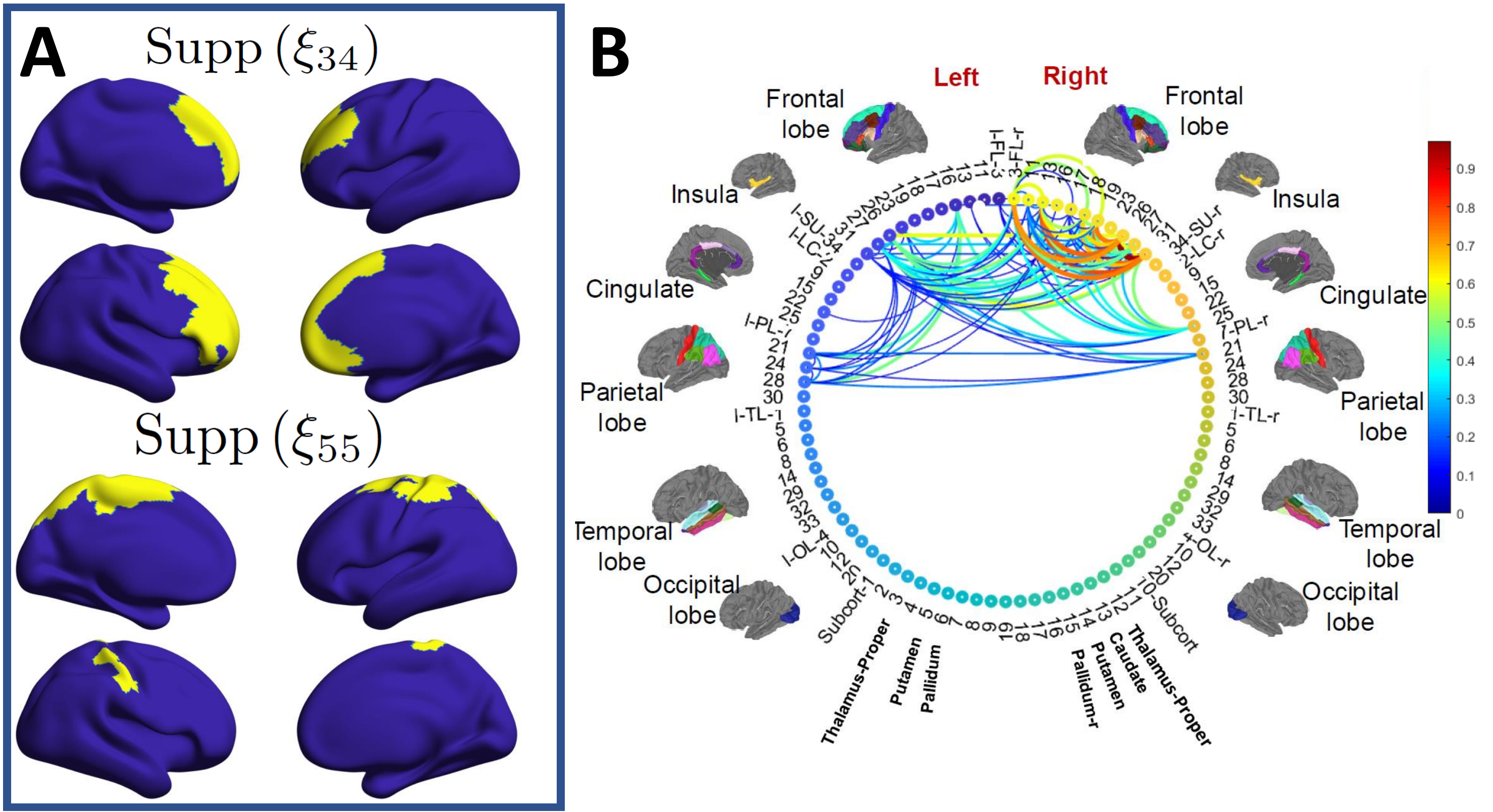}
    \caption{A: Support sets (yellow) of the two basis functions found to be significantly associated with connectivity differences in the delay discounting groups. B: Top $50\%$ of the edges in $\boldsymbol{A}^{CC}$.}
    \label{fig:subnet_select}
\end{figure}
\par 
The k=$34$ and k=$55$ embedding dimensions were found to be significantly different between steep and low discounting groups. Figure~\ref{fig:subnet_select} A. shows the non-zero support sets: $\text{Supp}(\xi_k):=\{\omega\in\Omega:\xi_k(\omega)\neq 0\}$, of the selected basis functions (yellow) plotted on the cortical surface. To visualize the associated connectivity patterns, we coarsened the continuous support sets using integration over the Desikan atlas \citep{desikan2006}. Specially, we define the $(a,b)$-th element in the coarsened adjacency matrix as
 $$
    \boldsymbol{A}_{ab}^{CC} := |E_{a}|^{-1}|E_{b}|^{-1}\int_{E_{a}}\int_{E_{b}}\sum_{k\in\{34,55\}}\mathbb{I}\{\xi_{k}(\omega_1)\xi_{k}(\omega_2) \neq 0\}d\omega_1d\omega_2
$$
where $E_{a},E_{b}\in\Omega$ are any two parcels from the Desikan atlas, $\mathbb{I}$ is the indicator function and $|E_{a}|$ denotes the surface measure. Figure~\ref{fig:subnet_select} B. provides a circular network plot showing the top $50\%$ of the edges in $\boldsymbol{A}_{ab}^{CC}$. Together, these views indicate the affected brain regions and connectivity patterns are predominantly within and between areas in the left and right frontal lobes and the left parietal lobe. Previous studies have shown these areas to be strongly related to delay discounting \citep{olson2009,owens2017,sebastian2014}.
\par\bigskip
\noindent{\textbf{Conclusion}:
This work introduces a novel modeling framework for structural connectivity using latent random functions defined on the product space of the cortical surface. To facilitate tractable representation and analysis, we formulate a data-adaptive reduced-rank function space that avoids the curse of dimensionality while also retaining a reasonable convergence rate. To form a basis for this space, we derive a penalized optimization problem and propose a novel computationally efficient algorithm for estimation. Experiments on real data from the HCP demonstrate that our method is able to produce more powerful representations for group-wise inference from SC data compared with traditional atlas-based approaches and can be used to localize group differences to brain regions and connectivity patterns on the cortical surface.}

\bibliographystyle{plainnat}
\bibliography{references}

\end{document}